\documentclass[a4paper,nofootinbib,showpacs,preprintnumbers,floatfix,superscriptaddress,pra,twocolumn]{revtex4}
\usepackage{amsmath, amsthm, amssymb}
\usepackage{graphicx}
\usepackage{dcolumn}
\usepackage{bm}
\usepackage{color}
\usepackage{amsmath}
\usepackage{amsfonts}
\usepackage{amssymb}
\usepackage{hyperref}
\usepackage{bbm}

\definecolor{myurlcolor}{rgb}{0,0,0.7}
\definecolor{myrefcolor}{rgb}{0.8,0,0}

\hypersetup{colorlinks,
linkcolor=myrefcolor,
citecolor=myurlcolor,
urlcolor=myurlcolor}

\newcommand{\R}{\mathbb{R}}

\newcommand{\C}{\mathbb{C}}

\newcommand{\ket}[1]{| #1 \rangle}

\newtheorem{prop}{Proposition}

\newtheorem{theorem}[prop]{Theorem}
\newtheorem{definition}[prop]{Definition}

\newcommand{\beq}{\begin{equation}}
\newcommand{\eeq}{\end{equation}}
\newcommand{\bea}[1]{\begin{equation}\begin{array}{#1}}
\newcommand{\eea}{\end{array}\end{equation}}
\newcommand{\beqn}{\begin{eqnarray}}
\newcommand{\eeqn}{\end{eqnarray}}

\renewcommand{\rho}{\varrho}

\newcommand{\processlist}[3][\relax]{%
  \def\listfinish{#1}%
  \long\def\listact{#2}%
  \processnext#3\listfinish}
\newcommand{\processnext}[1]{%
  \ifx\listfinish#1\empty\else\listact{#1}\expandafter\processnext\fi}

\renewcommand{\H}[1]{H(\processlist{X_}{#1})}

\begin{document}
\title{An entropic approach to local realism and noncontextuality}
\author{Rafael Chaves}
\affiliation{ICFO--Institut de Ciencies Fotoniques,
Mediterranean Technology Park, 08860 Castelldefels (Barcelona),
Spain}
\author{Tobias Fritz}
\affiliation{ICFO--Institut de Ciencies Fotoniques,
Mediterranean Technology Park, 08860 Castelldefels (Barcelona),
Spain}

\begin{abstract}
For any Bell locality scenario (or Kochen-Specker noncontextuality
scenario), the joint Shannon entropies of local (or noncontextual)
models define a convex cone for which the non-trivial facets are
tight entropic Bell (or contextuality) inequalities. In this paper
we explore this entropic approach and derive tight entropic
inequalities for various scenarios. One advantage of entropic
inequalities is that they easily adapt to situations like
bilocality scenarios, which have additional independence
requirements that are non-linear on the level of probabilities,
but linear on the level of entropies. Another advantage is that,
despite the nonlinearity, taking detection inefficiencies into
account turns out to be very simple. When joint measurements are
conducted by a single detector only, the detector efficiency for
witnessing quantum contextuality can be arbitrarily low.
\end{abstract}

 \maketitle
\section{Introduction}
\label{introduction}Quantum mechanics predicts that experiments
performed by space-like separated and independent observers may
display nonlocal correlations, which cannot be explained solely by
past interactions. The assumption that physical quantities have
well established values previous to any measurement and that
signals cannot travel faster than the speed of light, as
stipulated by special relativity, entails limits on the
correlations the observers may obtain. Such restrictions, usually
expressed as \emph{Bell inequalities}~\cite{Bell}, may be
surpassed within quantum theory when the observers share entangled
quantum states, and it is in this sense that the quantum
correlations are nonlocal.

Similarly, noncontextuality is a classically well-defined property
of mutually compatible observables. Two observables $A$ and $B$
are mutually compatible if the result for the measurement of $A$,
even if not performed, does not depend on the prior or
simultaneous measurement of $B$ and vice versa. The notion of
noncontextuality is precisely captured by the Kochen-Specker (KS)
theorem \cite{KStheorem}, stating that no noncontextual hidden
variable model (NCHV) can reproduce the results of quantum
mechanics. Interestingly, as opposed to Bell's theorem, the KS
theorem is not state-dependent and it holds for any physical
system (composite or not) with a state space of dimension higher than two.

Focusing on the Bell scenario, once the particular scenario is
defined---the number of spatially separated parties, the number of
measurements settings for each party and the number of outcomes
for each setting---the associated local (realistic) models form a
convex set with a finite number of extremal points, an object
known as the \emph{local polytope}~\cite{Pitowsky}. The tight Bell
inequalities are the non-trivial facets bounding the local
polytope. Given this geometric picture, a \emph{nonlocal model} is
a point outside the local polytope, or, equivalently, a point
which violates a Bell inequality.

It is clear that to properly understand nonlocality, it should be
considered from as many aspects as possible. In the
information-theoretic approach introduced by Braunstein and
Caves~\cite{CHSHentropic}, it was shown that if local realism
holds, then the joint Shannon entropies carried by the
measurements on two distant systems must satisfy certain
inequalities, which can be regarded as \emph{entropic Bell
inequalities}. One advantage of this entropic approach is that the
inequalities do not depend on the number of outcomes of the
measured observables, which implies that they can readily be
applied to quantum systems of arbitrary local dimension and to the
consideration of detection inefficiencies~\cite{Eberhard}.

In this paper, our aim is to further develop this entropic program
to Bell inequalities and also introduce entropic inequalities to
the study of (quantum and post-quantum) contextuality. As we
discuss in Sec.~\ref{sec:entropic_polytope} and in more detail
in~\cite{entropicpaper2}, the standard information-theoretic
inequalities for joint Shannon entropies (monotonicity and
submodularity) define a convex cone~\cite{Yeung} whose projection
to the joint entropies of jointly measurable observables is
another convex cone, whose facets correspond to all the optimal
Shannon-type entropic Bell inequalities.

From this geometrical point of view, we investigate the family of
chained entropic inequalities derived in
Ref.~\cite{CHSHentropic}---which we have shown to be the tightest
entropic inequalities in the appropriate
scenarios~\cite{entropicpaper2}---and look for violations in
quantum and post-quantum probabilistic theories. In particular, we
show how an entropic contextuality inequality can be
violated---for joint measurements of the commuting
observables---even for arbitrarily low detection efficiencies.
Moreover, in Sec.~\ref{sec:distilation} we analyze a possible
conection between the entropic CHSH inequality and nonlocality
distillation~\cite{Dist_Short, Dist_Foster, Dist_Brunner,
Dist_Allcock, Dist_Hoyer, Dist_Cavalcanti}. For a bilocality
scenario~\cite{swapping2} (a Bell scenario with two indepedent
sources which allow entanglement swapping~\cite{swapping}), we
derive in Sec.~\ref{sec:entropic_bilocal} all the relevant tight
(Shannon-type) entropic inequalities.

\section{Marginal Scenarios and Entropic Inequalities}
\label{sec:entropic_polytope} In the following, we collect the
basic definitions and results concerning entropic inequalities.
For a more detailed discussion, we refer to
Ref.~\cite{entropicpaper2}.

\begin{figure}[t!]
\begin{center}
\includegraphics[width=\linewidth]{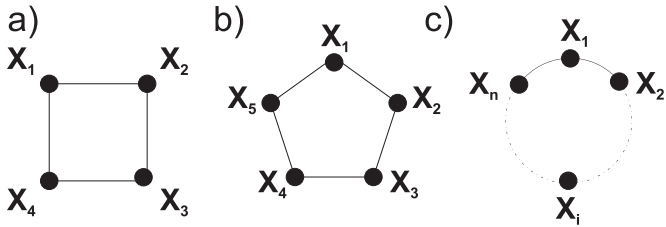}
\caption{(Color online) \label{polygons}
 {Some contextuality and nonlocality
scenarios; the number of outcomes of each observable is arbitrary.} {\bf a.} CHSH scenario, 2 parties with 2 measurements
settings each {\bf b.} Klyachko scenario, 5 observables arranged in a
cyclic configuration such that each observable is compatible with its
neighbors {\bf c.} Generalization of the Klyachko scenario, with
$n$ observables in a cyclic configuration (the $n$-cycle~\cite{LSW}). The unique non-trivial
entropic inequality in all these cases is given by~(\ref{polygon}). \label{fig:context}}
\end{center}
\end{figure}

\paragraph*{Marginal scenarios.} The concept of \emph{marginal scenario} subsumes both Bell scenarios and contextuality scenarios.
A marginal scenario is defined by specifying a set of observables
$X_1, \dots, X_{n}$ for which certain subsets are known to be
compatible and can be jointly measured. If a subset of these
observables can be jointly measured, then so can any smaller
subset; therefore, the collection of subsets of jointly measurable
observables should be closed under taking smaller subsets. These
are the \textit{marginal scenarios} which we have discussed
in~\cite{entropicpaper2}; see there for more background and
references on the marginal problem. The \textit{measurement
covers} of~\cite{AbramBrand} capture precisely the same idea in a
slightly different way.

\begin{definition}
A \emph{marginal scenario} $\mathcal{M}$ is a collection
$\mathcal{M}=\{S_1,\ldots,S_{|\mathcal{M}|}\}$ of subsets $S_i\subseteq\{X_1,\ldots,X_n\}$ such that
if $S\in\mathcal{M}$ and $S'\subseteq S$, then also $S'\in\mathcal{M}$.
\end{definition}

In this sense, every Bell scenario is a marginal
scenario~\cite[Sec.~2.4.1]{AbramBrand}: the collection of
observables $X_1,\ldots,X_{n}$ should comprise all observables of
all parties, and a subset of these observables is jointly
measurable if it does not contain two different observables of the
same party. For example in the bipartite case, with Alice having
access to observables $A_0,\ldots,A_{m-1}$ and Bob to
$B_0,\ldots,B_{m-1}$; if we write $X_i=A_{i-1}$ and
$X_{i+m}=B_{i-1}$ for $i=1,\ldots,m$, then the set of all
observables is $\{X_1,\ldots,X_{2m}\}$, and the subsets of jointly
measurable observables are the empty subset, the one-observable
subsets, and the two-observable subsets $\{X_i,X_j\}$ where $i\leq
m$ and $j>m$.

In a physical realization of a marginal scenario $\mathcal{M}$,
one measures some joint statistics for every $S\in\mathcal{M}$.
This means that one assigns a joint probability distribution to
every jointly measurable set of observables. We use notation like
$P(1,0|X_3,X_5)$ for the probability of obtaining the outcomes
$X_3=1$ and $X_5=0$ in a joint measurement of $X_3$ and $X_5$
(assuming that $\{X_3,X_5\}\in\mathcal{M}$).

If $S\in\mathcal{M}$ and $S'\subseteq S$, then one can take the
marginal of the distribution assigned to $S$ and obtain a
distribution over the outcomes of the observables in $S'$.
Naturally, this marginalized distribution should be the one
assigned to $S'$. Requiring this property leads to the
\textit{marginal models} of~\cite{entropicpaper2}, or,
equivalently, to the \textit{sheaf condition} and the
\textit{empirical models} of~\cite{AbramBrand}. In the case of
Bell scenarios, the marginal models are precisely the no-signaling
boxes.

\paragraph*{Noncontextual hidden variables.}

We now define when a marginal model is contextual. The most
intuitive theories of physics are those where there exists a
certain ``hidden'' variable $\lambda$, distributed according to
probabilities $\rho(\lambda)\geq 0$ with
$\sum_\lambda\rho(\lambda)=1$, such that $\lambda$ determines the
complete future behavior of the system. Here, completeness means
that the distribution $P(x|X_i,\lambda)$ of any observable $X_i$,
given a certain value of $\lambda$, should be independent of the
outcome distributions of all other observables. This implies that
when $X_i$ and $X_j$ are jointly measurable, then their outcome
distribution is given by

\beq \label{noncontextual} P(x_i,x_j|X_i,X_j)=\sum_{\lambda}
\rho(\lambda) P(x_i|X_i,\lambda)P(x_j|X_j,\lambda), \eeq and
similarly for cases where more than two observables are jointly
measured. If there exist conditional distributions
$P(x_i|X_i,\lambda)$ and a hidden variable distribution
$\rho(\lambda)$ such that~(\ref{noncontextual}) holds for all
jointly measurable pairs $\{X_i,X_j\}\in\mathcal{M}$ and more
generally for all $S\in\mathcal{M}$, then we say that we have
found a \emph{noncontextual hidden variable model}, and the given
marginal model $P$ is \textit{noncontextual}; otherwise $P$ is
called \emph{contextual}. In the case of Bell scenarios, the
noncontextual hidden variable models are precisely the local
hidden variable models, in which case we also use the standard
terminology of ``local'' and ``nonlocal''.

Following~\cite[Thm.~6]{LSW} or~\cite[Thm.~8.1]{AbramBrand}, we
note that the noncontextuality of $P$ is equivalent to the
existence of a joint distribution \beq \label{jointpd}
P(x_1,\ldots,x_{n}|X_1,\ldots, X_{n}) = p(x_1,\ldots,x_n) \eeq
which marginalizes to the given distributions for all
$S\in\mathcal{M}$.

The main question is: how is it possible to decide whether a given
marginal model $P$ in a marginal scenario $\mathcal{M}$ is
contextual or noncontextual?

\paragraph*{Entropic inequalities.}
From a joint probability distribution~(\ref{jointpd}), one can define the
associated Shannon entropy
\begin{align*}
 H(& X_1 \dots X_{n}) \\
& = - \sum_{x_1,\ldots,x_{n}} p(x,\ldots,x_{n})\log_{2} p(x_1,\ldots,x_{n})
\end{align*}
More generally, marginalizing the joint distribution to any subset
$S\in\mathcal{M}$ of the observables gives a joint entropy
$H(X_S)$, where we write $X_S$ for the tuple of observables
$(X_i)_{i\in S}$. This joint entropy $H(X_S)$ is also defined in
any marginal model, since the distribution of $X_S$ is known for
$S\in\mathcal{M}$.

As has first been noticed in~\cite{CHSHentropic} and as we have
developed more formally in a general
framework~\cite{entropicpaper2}, the noncontextuality of $P$,
i.e.~the existence of a joint distribution~(\ref{jointpd}),
implies that the $H(X_S)$, for $S\in\mathcal{M}$, satisfy certain
inequalities which may be violated in some contextual models.

\begin{definition}
\label{defentineq}
An \emph{entropic contextuality inequality} is a linear inequality
in the $H(X_S)$ for $S\in\mathcal{M}$ which is satisfied whenever
$P$ is noncontextual. In the special case of a Bell scenario, we
use the term \emph{entropic Bell inequality}.
\end{definition}

If some marginal model violates a certain entropic contextuality
inequality, then this inequality has witnessed the contextuality
of the marginal model.

In Ref.~\cite{entropicpaper2}, we have classified the entropic
contextuality inequalities in the $n$-cycle marginal scenarios.
This family of scenarios is defined by starting with any number
$n\geq 3$ of observables $X_1,\ldots,X_{n}$ and assuming that
$X_i$ and $X_{i+1}$ are pairwise jointly measurable for all
$i=1,\ldots,n$, where we write $X_{n+1}=X_1$ for ease of notation.
No other pairs of observables are assumed jointly measurable, and
no triples of observables are assumed jointly measurable. For
$n=4$, Fig.~\ref{fig:context}a shows that this can be identified
with the CHSH Bell scenario~\cite{CHSH}. For $n=5$ (see
Fig.~\ref{fig:context}b) it is the marginal scenario considered by
Klyachko~\cite{KlyMarg,KSKlyachko}, and hence we call it the
\textit{Klyachko scenario}. For general $n$, it can be visualized
as an $n$-sided polygon (Fig.~\ref{fig:context}c). Our result
in~\cite{entropicpaper2} is that the inequalities derived
in~\cite{CHSHentropic} are a complete set of tight entropic
inequalities in these scenarios. (An entropic inequality is tight
when no other entropic inequality can be strictly better than this
one, so that a complete set of tight entropic inequalities
completely characterizes the region of noncontextual marginal
models in entropy space.) Stated more formally:

\begin{theorem}[\cite{entropicpaper2}]
\label{polygonthm}
A marginal model in this scenario is entropically noncontextual
if and only if the entropic inequality
\beq \label{polygon} \H{{i}{i+1}} \: + \sum_{j\neq\, i,\,i+1} \H{j} \: \leq \sum_{j\neq i} \H{{j}{j+1}}
\eeq 
holds for all $i=1,\ldots,n$.
\end{theorem}

In principle, one may also want to consider inequalities
containing derived entropic quantities like conditional entropies
and mutual information. However, since these derived quantities
are themselves nothing but linear combinations of joint entropies,
every entropic inequality containing the former can be rewritten
in terms of the latter. In fact, as there are no linear relations
between joint entropies, every linear entropic inequality turns
into a unique normal form when expressed in terms of joint
entropies~\cite[Sec.~13.2]{YeungBook}.

Besides this, the relevance of our general
framework~\cite{entropicpaper2} lies in the fact that the standard
computational geometry methods like Fourier-Motzkin elimination,
which have been extensively used to characterize tight Bell
inequalities in probability space, can also be applied to derive
tight (Shannon-type) entropic contextuality inequalities and tight
(Shannon-type) entropic Bell inequalities. Unfortunately, these
computations are very demanding: applying the methods described
in~\cite{entropicpaper2} to the tripartite Bell scenario with two
observables per party, we have not been able to fully characterize
the tight (Shannon-type) entropic Bell inequalities in this
scenario.

In the following, we analyze the case $n=4$ (the CHSH scenario)
and the case $n=5$ (the Klyachko scenario) in some more detail and
analyze in particular their violation by marginals models arising
from quantum theory.

\subsection{The CHSH scenario}
\label{subsec:entropic_CHSH}

In the usual CHSH scenario~\cite{CHSH}, there are two distant
parties $A$ and $B$, each measuring one of two observables
$A_{0},A_{1}$ and $B_{0},B_{1}$, respectively. Each of these
observables is taken to have two possible outcomes, so that the
set of outcomes can be taken to be $\{-1,+1\}$. As usual for Bell
scenarios, we take observables to be jointly measurable when they
belong to different parties. In our framework, this can be
described by the marginal scenario where the collection of jointly
measurable sets of observables is given by
\begin{align*}
& \emptyset, \: \{A_0\} ,\:\{A_1\} ,\: \{B_0\} ,\: \{B_1\} , \\
&\{A_0,B_0\} ,\: \{A_0,B_1\} ,\: \{A_1,B_0\} ,\: \{A_1,B_1\} .
\end{align*}
This is illustrated in Fig.~\ref{polygons}a. The CHSH inequality
\cite{CHSH}
\begin{align}
\begin{split}
\label{corrCHSH}
& CHSH \\
&\phantom{x} =\langle A_0B_0\rangle+\langle A_0B_1\rangle+\langle
A_1B_0\rangle-\langle A_1B_1\rangle  \leq 2
\end{split}
\end{align}
together with its equivalent variants is a necessary and
sufficient condition for noncontextuality (i.e.~Bell locality) in
this scenario.

By taking $n=4$ and renaming the observables occuring
in~(\ref{polygon}), we obtain the entropic
inequality~\cite{CHSHentropic}
\begin{align}
\begin{split}
\label{entCHSH}
H(A_1 B_1) & + H(A_0) + H(B_0) \\
& \leq  H(A_0 B_0) + H(A_0 B_1) + H(A_1 B_0) .
\end{split}
\end{align}
In order to emphasize the similarity with~(\ref{corrCHSH}), we can rewrite this in terms of mutual information as
\begin{align*}
I(A_0:B_0) + I(A_0:B_1) + I(A_1:&B_0)  - I(A_1:B_1) \\
 & - H(A_0) - H(B_0) \leq 0 .
\end{align*}
The last two terms in the left-hand side are analogous to the
classical bound of $2$ in~(\ref{corrCHSH}). We abbreviate the
left-hand side by $CHSH_E$, so that the \textit{entropic CHSH
inequality} is
$$
CHSH_E \leq 0 .
$$
Note that the entropic CHSH inequality does not only apply to $\pm
1$-valued observables $A_x$ and $B_y$; it is completely irrelevant
how many outcomes each observable has---as long as the Shannon
entropies converge---and which outcomes these are.

\begin{figure} [t!]
\begin{center}
\includegraphics[width=0.8\linewidth]{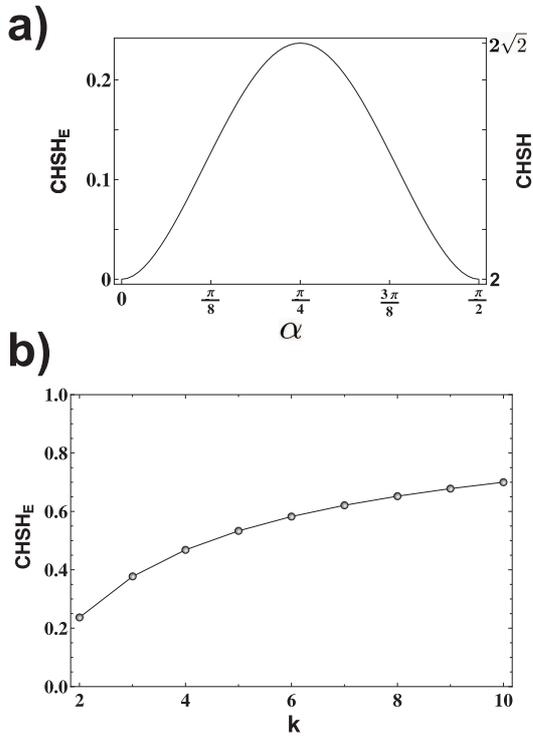}
\caption{(Color online) {\bf a.} Violation profile of the entropic
CHSH inequality for a pure state of the form $\cos \alpha
\ket{00}+\sin \alpha\ket{11}$. The violation profile essentially
coincides with the one of the standard CHSH inequality, although
the scales are different. The curve is obtained by optimizing over
the observables for each $\alpha$. {\bf b.} Maximal violation of
the entropic inequality (\ref{polygon}), achieved on a maximally
entangled state of two qubits ($\alpha= \pi /4$) and $k=2,...,10$
measurement settings for each party. \label{fig:violationprofile}}
\end{center}
\end{figure}

\paragraph*{Quantum violations of $CHSH_E$.}

We now assume that the two parties $A$ and $B$ share a pure
entangled two-qubit state. For any such state, one can choose the
local bases such that it has the form \beq \label{twoqubitstate}
|\psi\rangle = \cos \alpha \ket{00}+\sin \alpha\ket{11}, \eeq for
some $\alpha\in(0,\pi/2)$. Writing $A_x$ and $B_y$ also for the
$\pm 1$-valued quantum observables measured on $|\psi\rangle$, we
obtain the joint distributions, written in terms of the standard
notation of conditional probabilities,
\begin{align}
\begin{split}
\label{quantumCHSH}
P(a&,b|x,y) \\
& = \left\langle\psi\bigg|\left(\frac{1+(-1)^a A_x}{2}\otimes
\frac{1+(-1)^b B_y}{2}\right) \bigg|\psi\right\rangle .
\end{split}
\end{align}
For ease of later notation, we stipulate that the observables $A_x$ and $B_y$ are $\pm 1$-valued, but the outcomes $a$, $b$ are $\{0,1\}$-valued.

As shown in~\cite{CHSHentropic}, quantum correlations of the
form~(\ref{quantumCHSH}) do indeed lead to violations
of~(\ref{entCHSH}). These violations witness the non-existence of
local hidden variable models for~(\ref{quantumCHSH}). Numerical
optimization shows that the maximal violation of~(\ref{entCHSH})
for quantum correlations of the above form is achieved when all
measurement settings lie in the $Y$-$Z$ plane of the Bloch sphere,
that is, for measurement operators $A_{x}$ and $B_{y}$ of the form
$\sin\theta \cdot \sigma_{y} + \cos\theta \cdot \sigma_{z}$. (One
could as well take them all to lie in the $X$-$Z$ plane; what is
important is that they lie in the \textit{same} plane.) The
maximal violation is obtained for the maximally entangled state
$\alpha= \pi /4$, on which one gets $CHSH_E\approx +0.237$.

For other values of $\alpha$, the maximal violation
of~(\ref{entCHSH}), when optimized over the measurements, follows
the exact same profile as for the standard correlator
inequality~(\ref{corrCHSH}) as can be seen in
Fig.~\ref{fig:violationprofile}a. However, the measurements that
maximize the violation of $CHSH$ are not the ones which give the
maximal violation of $CHSH_E$. In fact, we will see below that
those choices of observables which produce the maximal $CHSH$
value for a certain $\alpha$ do not violate~(\ref{entCHSH}). In
general, for the standard CHSH scenario, the violation of the
standard inequality~(\ref{corrCHSH}) is a necessary but not
sufficient condition for the violation of~(\ref{entCHSH}).

We have also considered the inequalities~(\ref{polygon}) for any
even $n=2k$ as entropic Bell inequalities as follows. When A
chooses between $\pm 1$-valued observables $A_1,\ldots,A_k$ and B
among $B_1,\ldots,B_k$, then~(\ref{polygon}) becomes applicable
upon taking $X_{2i}=B_i$ and $X_{2i-1}=A_i$.
Fig.~\ref{fig:violationprofile}b shows our numerical results for
the maximal violation of~(\ref{polygon}) on a two-qubit
state~(\ref{twoqubitstate}).

\paragraph*{No-signaling violations of $CHSH_E$.}

In the CHSH scenario, there is a special class of marginal models
known as \emph{isotropic boxes}, which we would now like to study.
To begin, the \emph{Popescu-Rohrlich box} (PR box)~\cite{PRbox} is
defined to be the marginal model
\begin{equation}
P^{\mathrm{PR}}(a,b|x,y)=\frac{1}{4}\left[  1+\left(  -1\right)  ^{a\oplus b\oplus xy}\right],
\label{Ppr}
\end{equation}
It is the unique marginal model which maximally
violates~(\ref{corrCHSH}). Similarly, the \emph{isotropic box}
with parameter $C\in[0,1]$ is defined to be the marginal model
\begin{equation}
P^{\mathrm{iso}}(a,b|x,y) = \frac{1}{4}\left[  1+ C\left(
-1\right)  ^{{a}\oplus {b}\oplus xy}\right]. \label{isodist}
\end{equation}
It corresponds to a probabilistic mixture of $P^{\mathrm{PR}}$
with weight $C$ and uniform white noise $P^w$ with weight $1-C$.

Equivalently, an isotropic box can be characterized by having
uniformly random marginals, that is
$$
\langle A_x\rangle=\langle B_y\rangle =0,
$$
together with
$$
\langle A_0B_0\rangle = \langle A_0B_1\rangle = \langle
A_1B_0\rangle = -\langle A_1B_1\rangle \geq 0.
$$
The parameter $C$ is determined from this by $C = \langle A_0B_0\rangle$. Its relation to the CHSH value of the box is
simply given by $CHSH(P^{\mathrm{iso}}) = 4C$.

Any marginal model in the CHSH scenario can be transformed into an
isotropic box through a local depolarization process, keeping the
$CHSH$ value~(\ref{corrCHSH}) invariant~\cite{Masanes}. Therefore,
for many purposes it is enough to consider isotropic boxes only.

Interestingly, no isotropic box violates the entropic CHSH
inequality. In particular, this applies to the PR-box, although it
maximally violates the standard CHSH inequality. One can
understand this by noting that entropy only probes the probability
values occurring in a distribution, but not which probability
values get assigned to which outcomes. This means that the PR-box
is, as far as the entropies are concerned, equivalent to the
marginal model $P^c$ describing classical correlations,
\begin{equation}
\bigskip P^{c}(a,b|x,y)=\frac{1}{4}\left[  1+\left(  -1\right)  ^{{a}\oplus {b}}\right] \label{Pc}.
\end{equation}
Entropic quantities cannot distinguish between the perfect
anti-correlation of $A_1$ and $B_1$ as it appears in
$P^{\mathrm{PR}}$, and the perfect correlation of $A_1$ and $B_1$
as in $P^c$. All joint entropies of the PR-box coincide with those
of $P^c$.

For two-outcome measurements, the maximal violation
of~(\ref{entCHSH}) is $+1$ for the following reason: any marginal
model with two-outcome measurements will satisfy $H(A_0)\leq
H(A_0B_0)$ and $H(B_0)\leq H(A_1B_0)$; similarly, $H(A_1B_1)\leq
1+H(B_1)\leq H(A_0B_1)+1$. Taking these inequalities together
shows that $CHSH_E\leq 1$ for any such marginal model in the CHSH
scenario. This bound on the violation can indeed by achieved by
the no-signaling box
$$
P^{\mathrm{max}} = \tfrac{1}{2}P^{PR}+\tfrac{1}{2} P^{c}  ,
$$
which is an equal mixture of the PR-box with classical
correlations. $P^{\mathrm{max}}$ can be understood as the
probabilistic model in which each of the three pairs $(A_0,B_0)$,
$(A_0,B_1)$ and $(A_1,B_0)$ displays perfect correlation, while
the fourth pair $(A_1,B_1)$ is uncorrelated; see
also~\cite[Prop.~4.3]{entropicpaper2}. Note that $P^{\max}$
achieves a value of $3$ on $CHSH$, and therefore does not have a
quantum-mechanical realization since it is beyond Tsirelson's
bound of $2\sqrt{2}$.

This example shows that a convex combination of two non-violating marginal models may violate an entropic contextuality inequality. This
highlights the strongly non-linear character of entropic inequalities (see also Fig.~\ref{fig:entropie_and_distillation}).

\paragraph*{Discussion.}

It is a basic feature of Shannon entropy that the entropy of a
probability distribution is invariant under permutations of the
sample space. From this point of view, we find it surprising that
the entropic inequality~(\ref{entCHSH}) can be violated at all.
Violations of entropic contextuality inequalities witness a very
particular kind of contextuality: if a probabilistic model
violates an entropic inequality, then so does every other
probabilistic model obtained by permuting the outcome
probabilities of a joint measurement, provided that the permuted
joint distribution has the same marginals. For example, this leads
to the phenomenon observed above that the PR-box $P^{\mathrm{PR}}$
is entropically indistinguishable from classical correlation
$P^c$.

Along similar lines, the only
symmetry operations that can be applied in order to transform an entropic contextuality inequality
into an equivalent one are permutations of the observables which
map jointly measurable sets to jointly measurable sets. In the
case of a Bell scenario, these permutations are either
permutations of the parties, permutations of the observables of
some party, or arbitrary combinations thereof. Relabelings of the
outcomes of an observable do not change the inequality, again due to the fact that
entropies are invariant under outcome permutations.

\subsection{The Klyachko scenario}
\label{subsec:entropic_context}

A very simple state-dependent proof of the Kochen-Specker theorem
with only five two-outcome observables was given by Klyachko et
al.~in~\cite{KlyMarg,KSKlyachko}. The marginal scenario in this
case is the one depicted in figure~\ref{fig:context}b: there are
five $\pm 1$-valued observables $X_1,X_2,X_3,X_4,X_5$ such that
$X_{i}$ and $X_{i+1}$ (modulo 5) are compatible.

The so-called Klyachko inequality is a necessary and sufficient
condition for noncontextuality in this
scenario~\cite{KlyMarg,KSKlyachko}. It is given by
\begin{equation}
K_{5}=\sum_{i=1}^5\langle X_iX_{i+1}\rangle\geq
-3\label{KlyachkoCorr}
\end{equation}

Due to Theorem~\ref{polygonthm}, the only non-trivial entropic
inequality in the Klyachko scenario is given by, up to cyclic
permutations of the observables, the \emph{entropic Klyachko
inequality}
\begin{align}
\begin{split}
\label{ke}
H&(X_1  X_5) + H(X_2) + H(X_3) + H(X_4) \\
 & - H(X_1X_2) - H(X_2X_3) - H(X_3X_4) - H(X_4X_5) \leq 0
\end{split}
\end{align}

To investigate the quantum violations of the entropic Klyachko
inequality, we choose two-outcome observables on $\C^3$ of the
form
\begin{equation}
X_{i}=2|v_{i}\rangle\langle v_i|-\mathbbm{1}
\end{equation}
with the vectors $|v_{i}\rangle$ given by
\begin{align*}
\begin{split}
\label{v04}
|v_{1}\rangle  & =\left(  0,0,1\right)  \\
|v_{2}\rangle   & =\left(  \sin\theta,\cos\theta,0\right)  \\
|v_{3}\rangle   & =\mathcal{N}^{-1}\left(  \cos\theta\sin\phi,-\sin\theta\sin\phi,\sin\theta\cos\phi\right)  \\
|v_{4}\rangle   & =\left(  0,\cos\phi,\sin\phi\right)  \\
|v_{5}\rangle   & =\left(  1,0,0\right) ,
\end{split}
\end{align*}
where the normalization factor is
$\mathcal{N}=\sqrt{\sin^{2}\theta+\cos^{2}\theta\sin^{2}\phi}$. Up
to choice of basis and multiplying the $|v_i\rangle$ by irrelevant
phases, every configuration of $5$ unit vectors
$|v_1\rangle,\ldots,|v_5\rangle\in\R^3$ with $|v_i\rangle$
orthogonal to $|v_{i+1}\rangle$ is of this form.

Since each $|v_i\rangle$ is orthogonal to $|v_{i+1}\rangle$, the
observable $X_i$ commutes with $X_{i+1}$, so that these two
observables are compatible and we can talk about their joint
measurement. Also thanks to orthogonality, their joint outcome
$(X_i=1,X_{i+1}=1)$ never occurs.

We write $|v_i\times v_{i+1}\rangle$ for a unit vector orthogonal
to both $|v_i\rangle$ and $|v_{i+1}\rangle$.

Upon measuring these observables on some initial state
$|\psi\rangle\in\C^3$, the non-vanishing joint outcome
probabilities are given by
\begin{align}
\begin{split}
P(0,1|X_i,X_{i+1})  & = P(1|X_{i+1}) = |\langle v_{i+1}|\psi\rangle|^2 \\
P(1,0|X_i,X_{i+1})  & = P(1|X_i) =  |\langle v_{i}|\psi\rangle|^2 \\
P(0,0|X_i,X_{i+1})  & = 1 - P(1|X_i) - P(1|X_{i+1}) \\
 & = |\langle v_i\times v_{i+1}|\psi\rangle|^2 .
\end{split}
\end{align}

\paragraph*{Numerical results.}

Numerical calculations show that the maximal qutrit violation
of~(\ref{ke}) with observables $X_i$ occurs on a qutrit state of
the form
\begin{equation}
\label{kestate} \left\vert \psi\right\rangle =
\frac{1}{\sqrt{1+\sin^{2}\alpha}} \left(
\sin\alpha,\cos\alpha,\sin\alpha\right)
\end{equation}
with $\alpha\approx 0.29736$ and $\theta=\phi\approx 0.24131$, for
which the left-hand side of~(\ref{ke}) is $\approx +0.091$.

\paragraph*{Analytical proof of quantum violations.}

We would like to present an analytic proof showing that quantum
violations of~(\ref{ke}) occur with the $X_i$ and $|\psi\rangle$
of the above form. We set $\theta=\phi$ and $\alpha=2\phi$ and
expand everything for $\phi\ll 1$. Then, by symmetry, the joint
outcome distribution of $X_4$ and $X_5$ coincides with the one of
$X_2$ and $X_1$; likewise, the joint outcome distribution of $X_3$
and $X_4$ coincides with the one of $X_3$ and $X_2$; the remaining
joint distributions, up to $O(\phi^3)$, are listed in
Table~\ref{keav}.

With this symmetry,~(\ref{ke}) is equivalent to
\begin{align}
\begin{split}
\label{kesym}
  H(X_1X_5) & + 2H(X_2) + H(X_3) \\
& - 2H(X_1X_2) - 2H(X_2X_3) \leq 0 .
\end{split}
\end{align}

\begin{table}
\begin{tabular}{c|cccc}
 & $(0,1)$ & $(1,0)$ & $(0,0)$ \\
\hline\\
$X_1,X_5$ & $4\phi^2$ & $4\phi^2$ & $1-8\phi^2$  \\
$X_1,X_2$ & $1-5\phi^2$ & $4\phi^2$ & $\phi^2$ \\
$X_2,X_3$ & $\frac{9}{2}\phi^2$ & $1-5\phi^2$ &
$\frac{1}{2}\phi^2$
\end{tabular}
\caption{Joint outcome probabilities for the analytic violation.
All values are up to $O(\phi^3)$.} \label{keav}
\end{table}

The corresponding relevant entropy values are given by
\begin{align*}
H(X_1) &= - 4\phi^2\log(\phi^2) + O(\phi^2) \\
H(X_2) &= - 5\phi^2\log(\phi^2) + O(\phi^2)\\
H(X_3) &= - \frac{9}{2}\phi^2 \log(\phi^2) + O(\phi^2)\\
H(X_1X_2) &= - 5\phi^2\log(\phi^2) + O(\phi^2)\\
H(X_2X_3) &= - 5\phi^2\log(\phi^2) + O(\phi^2)\\
H(X_1X_5) &= - 8\phi^2\log(\phi^2) + O(\phi^2)\\
\end{align*}
With this, the left-hand side of~(\ref{kesym}) is
$-\frac{5}{2}\phi^2\log(\phi^2) + O(\phi^2)$, which is positive
for small enough $\phi$.

\paragraph*{Detection inefficiencies: single-detector model.}

One can take advantage of the fact that entropic inequalities can
handle any finite number of outcomes and use the same approach to
investigate the more realistic case with detection inefficiencies.
We will consider two possible scenarios: one where compatible
observables are measured jointly (one detector), and one where
compatible observables are measured sequentially (two detectors).

In the single-detector model with detection efficiency
$\eta\in[0,1]$, there is an additional outcome
$(\emptyset,\emptyset)$ which represents the no-click event of the
detector for each jointly measurable pair $(X_i,X_{i+1})$. The new
outcome probabilities $P^\eta$ are given by
\begin{align}
\begin{split}
P^\eta(x_i,x_{i+1}|X_i,X_{i+1})  & =\eta P(x_i,x_{i+1}|X_i,X_{i+1}) \\
P^\eta(\emptyset,\emptyset|X_i,X_{i+1})  & =1-\eta ,
\end{split}
\label{Psingledetector}
\end{align}
where the measurments are now $\{-1,+1,\emptyset\}$-valued, and
$x_i,x_{i+1}\in\{-1,+1\}$ are the ``proper'' outcomes. In this
model, a no-detection event always occurs for both observables
simultaneously.

The joint probabilities~(\ref{Psingledetector}) marginalize to the
single-observable distributions
\begin{align}
\begin{split}
P(x_i|X_i)  & =\eta P(x_i|X_i)\\
P(\emptyset|X_i)  & =1-\eta.
\end{split}
\label{singlemarginal}
\end{align}

\begin{prop}
With this model of inefficiencies,
\begin{align*}
H^\eta(X_i) &= \eta H(X_i) + h(\eta) \\
H^\eta(X_iX_{i+1}) &= \eta H(X_iX_{i+1}) + h(\eta).
\end{align*}
where $h(\eta) = -\eta\log\eta - (1-\eta)\log(1-\eta)$ is the
binary entropy.
\end{prop}

\begin{proof}
This follows from an application of the grouping rule of Shannon
entropy, see e.g.~\cite[Sec.~2.179]{CT}.
\end{proof}

Upon plugging these equations into~(\ref{ke}), one finds that the
contributions of $h(\eta)$ cancel, so that the left-hand side
simply scales as a linear function of $\eta$. Therefore, the
entropic contextuality inequality~(\ref{ke}) has violations for
any $\eta \geq 0$. Moreover, the maximal violation with qutrit
measurements of the form~(\ref{ke}),~(\ref{v04}) is given by the
same state and vectors $|v_i\rangle$ which maximizes the violation
for $\eta=0$ (which are $\alpha\approx 0.29736$ and
$\theta=\phi\approx 0.24131$).

\paragraph*{Detection inefficiencies: two-detector model.}

We now assume that the joint measurement of $X_i$ and $X_{i+1}$ is
realized by one detector measuring $X_i$ and another detector
measuring $X_{i+1}$. Again, we take each detector to have an
efficiency of $\eta\in[0,1]$, for simplicity the same value for
all $5$ detectors, such that the no-click event of the first
detector is independent of the no-click event of the second
detector. A physical situtation leading to this kind of model may
be a sequential scheme where the system passes through a
non-demolition measurement in the first detector before reaching
the second detector.

Consequently, the jointly measurable pair $(X_{i},X_{i+1})$ has an
outcome distribution, with $x_i\in\{-1,+1\}$,
\begin{align}
\begin{split}
P^\eta(x_i,x_{i+1}|X_i,X_{i+1})  & =\eta^{2}P(x_i,x_{i+1}|X_i,X_{i+1})\\
P^\eta(x_i,\emptyset|X_i,X_{i+1})  & =\left(  1-\eta\right)  \eta P(x_i|X_i)\\
P^\eta(\emptyset,x_{i+1}|X_i,X_{i+1})  & =\left(  1-\eta\right)  \eta P(x_{i+1}|X_{i+1})\\
P^\eta(\emptyset,\emptyset|X_i,X_{i+1})  & =\left(  1-\eta\right)
^{2}
\end{split}
\end{align}
which again marginalize to the single-observable
distributions~(\ref{singlemarginal}).

\begin{prop}
With this model of inefficiencies,
\begin{align}
\label{ineff2}
H^\eta(X_i) =& \eta H(X_i) + h(\eta) \\
H^\eta(X_iX_{i+1}) =& \eta^2 H(X_iX_{i+1})  \nonumber\\
 & + \eta(1-\eta) \left[ H(X_i) + H(X_{i+1}) \right] + h(\eta). \nonumber
\end{align}
where $h(\eta) = -\eta\log\eta - (1-\eta)\log(1-\eta)$ is the
binary entropy.
\end{prop}

\begin{proof}
Again, this follows from an application of the grouping rule of
Shannon entropy~\cite[Sec.~2.179]{CT}.
\end{proof}

Due to the additional terms in~(\ref{ineff2}), the required
detection efficiency for witnessing quantum violations in the
two-detector model turns out be very high, $\eta\approx 0.995$.

\section{Nonlocality distillation and the entropic CHSH inequality}
\label{sec:distilation}

\begin{figure}[t!]
\begin{center}
\includegraphics[width=0.6\linewidth]{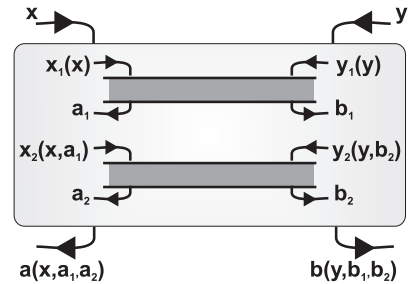}
\caption{(Color online) General nonlocality distillation protocol
with two boxes. The inputs for the second box may locally depend
on the output of the first one and the final output of the
combined boxes is a function of all inputs and outputs.
\label{fig:twoboxes}}
\end{center}
\end{figure}


As we have seen in the previous sections, finding violations of an
entropic contextuality inequality, for example of the entropic
CHSH inequality~(\ref{entCHSH}), is not easy, and there are few
quantum-mechanical models which do violate them. Therefore, we
regard the violation of an entropic contextuality inequality
(entropic Bell inequality) as a witness of a very strong form of
contextuality (nonlocality). So what does the violation of an
entropic contextuality inequality tell us about the violating
model? Can the violation be regarded as a resource for something?

Since Shannon entropy is an asymptotic quantity which measures the
effective size of a probability distribution on the level of many
copies (see for example the Asymptotic Equipartition
Property~\cite{YeungBook}), we also expect any answers to these
questions to be concerned with the limit of many copies.

In this section, we would like to consider the case of the CHSH
scenario and investigate a bit, on a purely phenomenological
level, one particular property which also has an asymptotic
flavor: the property of \emph{nonlocality distillation}. However,
we barely have a precise hypothesis to offer---let alone a
proof---and the similarities we will describe in the following may
easily turn out to be superficial.

Considering the CHSH scenario, it was recently shown that
nonlocality can be distilled \cite{Dist_Foster,
Dist_Brunner,Dist_Allcock, Dist_Hoyer, Dist_Cavalcanti}: by
locally processing several copies of certain bipartite
no-signaling boxes, one can increase the amount of nonlocality
according to a particular nonlocality measure. In the following,
we consider bipartite no-signaling boxes with binary inputs and
outputs. Each party can \emph{wire} boxes together using classical
circuitry to produce a new binary-input/binary-output box.

The general distillation protocol with two copies of a
no-signaling box is displayed in Fig. \ref{fig:twoboxes}. A
certain protocol distills the nonlocality if the no-signaling box
obtained after the wiring is more nonlocal---according to a
certain measure---than the original box. As a measure of
nonlocality, we will consider the \emph{nonlocal part} of the
\emph{EPR2 decomposition}~\cite{EPR2}, defined as follows. Any
no-signaling box $P$ can be decomposed into convex combinations of
a purely local box $P^L$ and another box $P^{NL}$,
\begin{equation}
P = (1-q) P^{L} + q P^{NL},
\end{equation}
with some coefficient $q\in[0,1]$. The minimal coefficient $q$,
when minimized over all such decompositions is the \emph{nonlocal
content} of $P$. By definition, computing the nonlocal content is
a linear program. In the case of the CHSH scenario, there is a
linear relationship between the nonlocal content and the violation
of the standard CHSH inequality~\cite{Dist_Cavalcanti}; therefore,
in the present context we may regard the $CHSH$ value as our
measure of nonlocality.

\begin{figure}[t!]
\begin{center}
\includegraphics[width=0.8\linewidth]{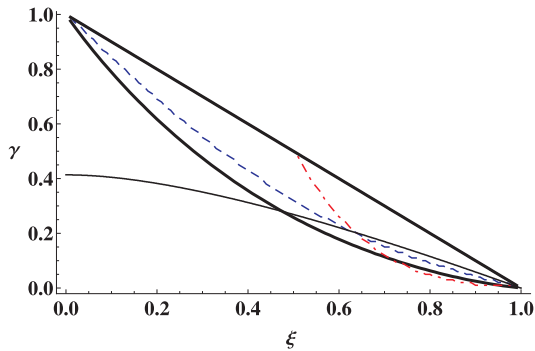}
\caption{(Color online) Slice of the no-signaling polytope
corresponding to the family of boxes~(\ref{prob_brunner}). The
quantum region and $CHSH_{E}$ violations are bounded respectively
by the thin black and thick black curves. The regions of
distillability by the wirings (\ref{wiring_foster}) and
(\ref{wiring_Cavalcanti}) correspond respectively to the
dashed-dotted red and dashed blue lines. The region of violations
of the entropic CHSH inequality is strictly larger than the region
of correlations distilable with the wiring
(\ref{wiring_Cavalcanti}). On the other hand, in a small part of
the quantum region, nonlocality can be distilled by the protocol
(\ref{wiring_foster}) while not being detected by $CHSH_{E}$.
\label{fig:entropie_and_distillation}}
\end{center}
\end{figure}

We now focus on two particular wirings. The first one was proposed
in \cite{Dist_Foster} and is given by 
\begin{equation}
\left\{
\begin{array}
[c]{c}
x_{1}=x,\: x_{2}=x,\: a=a_{1}\oplus a_{2}\\
y_{1}=y,\: y_{2}=y,\: b=b_{1}\oplus b_{2},
\end{array}
\right. \label{wiring_foster}
\end{equation}
while the second one was proposed in~\cite{Dist_Cavalcanti} and is
given by 
\begin{equation}
\left\{
\begin{array}
[c]{c}
x_{1}=x,x_{2}=x\oplus a_{1}\oplus1,a=a_{1}\oplus a_{2}\oplus1  \\
y_{1}=1,y_{2}=yb_{1},b=b_{1}\oplus b_{2}\oplus1.
\end{array}
\right. \label{wiring_Cavalcanti}
\end{equation}
As shown in~\cite{Dist_Cavalcanti}, these two wirings completely
characterize the distilability in the CHSH scenario on the
two-copy level.

For example, for the wiring~(\ref{wiring_foster}), the wired
no-signaling box is given by 
\begin{equation}
P'(a,b|x,y)  =
{\sum\limits_{\substack{a_{1}+a_{2}=a\\b_{1}+b_{2}=b}}} P\left(
a_{1},b_{1}|x,y\right)  P(a_{2},b_{2}|x,y),
\end{equation}
As for any wiring protocol with two copies of the original box,
the new box $P'$ is a quadratic function of the original one. So
the first (superficial) similarity between wirings and the
entropic CHSH inequality is nonlinearity.

Interestingly, in the case of isotropic boxes---which we found
in~\ref{subsec:entropic_CHSH} not to violate the entropic CHSH
inequality---distillation is not possible~\cite{Dist_Short}. More
generally, we have considered the family of boxes
\begin{equation}
P_{\gamma,\xi}=\gamma P^{PR}+\xi P^{c}+\left(1-\gamma-\xi\right)
P^{f}, \label{prob_brunner}
\end{equation}
with parameters $\gamma,\xi\in[0,1]$, $P^{PR}$ and $P^{c}$ given
respectively by~(\ref{Ppr}) and~(\ref{Pc}), as well as 
\begin{equation}
\bigskip P^{f}(a,b|x,y)=\frac{1}{8}\left[  2+\left(  -1\right)  ^{a\oplus b\oplus
xy}\right],
\end{equation}
which is $P^f=\tfrac{1}{2}P^{PR}+\tfrac{1}{2}P^{w}$, half-way
between the PR-box and white noise.

For $\gamma=0$, this is a family of boxes with $CHSH=0$. The
family~(\ref{prob_brunner}) forms a triangle in the no-signaling
polytope extending from the boundary of the local polytope up to
the PR-box. As displayed in
Fig.~\ref{fig:entropie_and_distillation}, the subset of boxes
violating $CHSH_E$ and the subset of boxes distillable
via~(\ref{wiring_Cavalcanti}) or~(\ref{wiring_foster}) are very
similar. We expect that by going beyond the two-copy level and
considering all possible wiring protocols, the subset of
distillable boxes will enlarge. It seems reasonable to ask whether
it will contain all boxes violating $CHSH_E$. In general, could
the violation of $CHSH_E$ be a sufficient condition for
distillability?

\begin{figure}[t!]
\begin{center}
\includegraphics[width=0.8\linewidth]{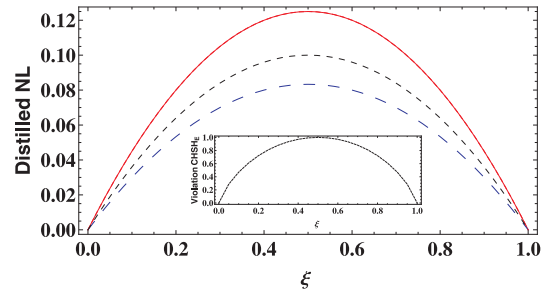}
\caption{(Color online) Distilled nonlocality (increase in
nonlocal content) of the family~(\ref{dfamily}) upon application
of the wiring (\ref{wiring_us}) for $d=2,4$ (red full line), $d=3$
(blue large dashed line) and $d=5$ (black dashed line) as computed
by linear programming. The inset shows the violation of
$CHSH_{E}$, which does not depend on $d$.
\label{fig:entropie_and_distillation2}}
\end{center}
\end{figure}

Moreover, we have investigated the bipartite scenario with two
$d$-outcome observables per party, to which $CHSH_{E}$ still
applies. We have considered the family of no-signaling boxes
\begin{equation}
\label{dfamily} P_{\xi}=\xi P^{PR}_{d}+(1-\xi) P^{c}_{d},
\end{equation}
where $P^{PR}_{d}$ is the generalized PR-box~\cite{BLMPSR}
\begin{equation}
P^{PR}_{d}(a,b|x,y) = \left\{
\begin{array}{cl}
1/d & \text{if }a-b\equiv xy \mod d\\
0 & \text{otherwise}%
\end{array}
\right. ,\label{Pprd}
\end{equation}
while $P^{c}_{d}$ is the classically correlated box
\begin{equation}
P^{c}_{d}(a,b|x,y) = \left\{
\begin{array}{cl}
1/d & \text{if }a=b\\
0 & \text{otherwise}%
\end{array}
\right. .\label{Pcd}
\end{equation}
It follows from the definition~(\ref{Pcd}) and an application of
the CGLMP inequality~\cite{CGLMP} that the nonlocal content of
$P_{\xi}$ is simply $\xi$. On the other hand, $CHSH_{E}=-\xi
\log{\xi}-(1-\xi) \log{1-\xi}$ turns out to be the binary entropy.

To probe the nonlocality distillation of $P_{\xi}$, we have
considered the wiring 
\begin{equation}
\left\{
\begin{array}[c]{c}
x_{1}=x,\: x_{2}=xa_{1}\operatorname{mod}2,\: a=(a_{1}+a_{2})\operatorname{mod}d\\
y_{1}=y,\: y_{2}=yb_{1}\operatorname{mod}2,\:
b=(b_{1}+b_{2})\operatorname{mod}d
\end{array}
\right. , \label{wiring_us}
\end{equation}
which can be regarded as a generalization of the wiring proposed
in Ref.~\cite{Dist_Brunner}.
Fig.~\ref{fig:entropie_and_distillation2} shows the increase of
nonlocal content achievable with this protocol for any
$d=2,\dots,5$ and the value of $CHSH_{E}$; we observe
qualitatively identical behavior.

\section{Entropic Bilocality inequalities}
\label{sec:entropic_bilocal}

\begin{figure}[t!]
\begin{center}
\includegraphics[width=0.6\linewidth]{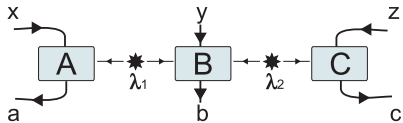}
\caption{(Color online) Swapping experiment scenario, where two
independent entangled states are prepared and sent to three
distant parts. If the central part B measures in a entangled
basis, the subsystems A and C, that had no previous contact,
become nonlocally correlated. \label{fig:swapping}}
\end{center}
\end{figure}

\paragraph*{Bilocality scenarios.}

In entanglement swapping~\cite{swapping}, there are two sources of
entangled quantum states. The first source (on the left in
Fig.~\ref{fig:swapping}) sends half of its entangled state to A
and the other half to B; the second source (on the right in
Fig.~\ref{fig:swapping}) sends half to B and half to C. If B
applies the right kind of entangled measurement between the two
quantum systems which he receives, then, conditioned on an outcome
of this measurement, the post-measurement state between A and C
will be entangled.

This idea has been used to obtain strong bounds~\cite{swapping2}
on the existence of local hidden variable models of entangled
quantum states augmented by the assumption that the hidden
variables $\lambda_1$ and $\lambda_2$ describing the two sources
are probabilistically independent; this is the \textit{bilocality}
assumption.

Under the assumption of local realism without the bilocality
assumption, the conditional probabilities for the scenario of
Fig.~(\ref{fig:swapping}) would have the form
\begin{align}
\begin{split}
& P\left(  a,b,c|x,y,z\right)  \\
& = {\displaystyle\sum\limits_{\lambda_{1},\lambda_{2}}} \rho(
\lambda_{1},\lambda_{2}) P(a|x,\lambda_{1})P(b|y,\lambda
_{1},\lambda_{2})P(c|z,\lambda_{2})
\end{split}
\end{align}
where $\rho(\lambda_1,\lambda_2)$ is the probability distribution
over the pairs of hidden variables. As usual,
$\left\{x,y,z\right\}$ and $\left\{a,b,c\right\}$ describe,
respectively, the inputs and outputs at each local part. Since the
two systems received by Bob are being jointly measured, they can
be treated as a single entity.

Since $\rho$ can be chosen such that $\lambda_1=\lambda_2$ occurs
with probability $1$, such a model is equivalent to one of the
form
\begin{align}
\begin{split}
\label{threelocal}
& P\left(  a,b,c|x,y,z\right) \\
& = {\displaystyle\sum\limits_{\lambda}} \rho\left( \lambda\right)
P(a|x,\lambda)P(b|y,\lambda)P(c|z,\lambda) ,
\end{split}
\end{align}
which is the standard description of local realism in a
three-party Bell scenario.

The bilocality assumption in addition imposes independence of
$\lambda_1$ and $\lambda_2$, so that the distribution $\rho$ is
required to factor as $ \rho\left( \lambda_{1},\lambda_{2}\right)
=\rho_1(\lambda_{1})  \rho_{2}(\lambda_{2})$, which means that a
\emph{bilocal model} is one of the form
\begin{align}
\begin{split}
\label{bilocaldecom}
& P\left(  a,b,c|x,y,z\right)  \\
& = {\displaystyle\sum\limits_{\lambda_{1},\lambda_{2}}}
\rho_1(\lambda_{1})\rho_2(\lambda_{2})
P(a|x,\lambda_{1})P(b|y,\lambda
_{1},\lambda_{2})P(c|z,\lambda_{2}) .
\end{split}
\end{align}
A direct calculation shows that every such model satisfies,
besides the usual no-signaling equations, also the condition \beq
\label{binosig} \sum_b P(a,b,c|x,y,z) = P(a|x)P(c|z) \quad\forall
a,c,x,y,z \eeq where $P(a|x)$ and $P(c|z)$ are the marginal
behaviors of A and C, respectively.

However, there are no-signaling boxes $P(a,b,c|x,y,z)$ which
satisfy~(\ref{binosig}), but nevertheless are not bilocal,
i.e.~cannot be written in the form~(\ref{bilocaldecom}). Many
examples of this form are even local boxes of the
form~(\ref{threelocal}). We will soon see more concrete examples
of this in which the ``non-bilocality'' can be witnessed by an
entropic inequality.

\paragraph*{Entropic bilocality inequalities.}

Due to the nonlinearity of the bilocality condition
$\rho(\lambda_1,\lambda_2)=\rho_1(\lambda_1)\rho_2(\lambda_2)$,
the set of $P(a,b,c|x,y,z)$ of the form~(\ref{bilocaldecom}) is
not convex~\cite{swapping2}, and it is difficult to determine
whether a given $P$ lies in this bilocal set or not. This is where
entropic inequalities enter: as we will show in the following,
they give necessary requirements for $P$ to be bilocal in terms of
inequalities \textit{linear} in the entropies of $P$. This
linearity is already visible on the level of the two sources,
whose probabilistic independence is equivalent to vanishing mutual
information, $I(\lambda_{1}:\lambda_{2})=0$, which is a linear
entropic equation
$$
H(\lambda_1\lambda_2) = H(\lambda_1) + H(\lambda_2) .
$$

For the sake of concreteness, we consider the specific scenario
where A and C have two available measurement settings to choose
from, whereas B always applies the same fixed measurement, so that
$x,z\in\{0,1\}$ and $y=0$. This corresponds to $5$ observables
$A_0,A_1,B,C_0,C_1$. A subset of these $5$ observables is jointly
measurable whenever it contains at most one observable of A and at
most one of B. Following~\cite{entropicpaper2}, one can visualize
the marginal scenario as in Fig.~\ref{simpcomplex}. In addition,
the observables $A_0$ and $A_1$ are independent of the observables
$C_0$ and $C_1$; therefore, in order to calculate the entropic
inequalities for this scenario following the procedure
of~\cite{entropicpaper2}, we need to consider the independence
constraint \beq \label{bilocalci} I(A_0A_1:C_0C_1) = 0, \eeq which
we write out in terms of joint entropies as \beq
\label{bilocalent} H(A_0A_1C_0C_1) = H(A_0A_1) + H(C_0C_1) . \eeq
Note that the data processing inequality implies that the
independence constraints
\begin{align}
\begin{split}
\label{bilocalci2}
I(A_x:C_0C_1) & = 0,\\
I(A_0A_1:C_z) & = 0,\\
I(A_x:C_z) & =0,
\end{split}
\end{align}
follow from~(\ref{bilocalci}).

\begin{figure}[t!]
\begin{center}
\includegraphics[width=0.4\linewidth]{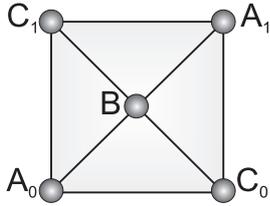}
\caption{The simplicial complex of jointly measurable observables
in the bilocality scenario discussed in the main text. Two
observables are jointly measurable whenever they share an edge;
three observables are jointly measurable whenever they are the
vertices of one of the four triangles.} \label{simpcomplex}
\end{center}
\end{figure}

We have used the computational approach of~\cite{entropicpaper2},
augmented by the independence constraints~(\ref{bilocalent})
and~(\ref{bilocalci2}) written out in terms of joint entropies, in
order to calculate all the (Shannon-type) entropic inequalities in
this bilocality scenario; although including the
constraints~(\ref{bilocalci2}) would not have been strictly
necessary, it helps in speeding up the computation. This
computation has resulted in $4$ equations and $52$ tight
inequalities. The $4$ equations are precisely the independence
conditions $I(A_x:C_z) = 0$. The $52$ inequalities fall into $10$
symmetry classes which we have listed in
Table~\ref{bilocalineqstable}. The first four inequalities are
trivial in the sense that they will hold for any no-signaling box
$P(a,b,c|x,y,z)$ satisfying~(\ref{binosig}), while the other six
inequalities can potentially be violated by non-bilocal boxes.

\begin{table*}
\begin{tabular}{|c|c c c c| c c c c| c c c c| c c c c| c
c c c| c c c c|} \hline
\multicolumn{25}{|c|}{Entropic Bilocality Inequalities}\\
\hline
      \multicolumn{1}{|c|}{\textbf{\#}}
    & \multicolumn{4}{|c|}{\scriptsize$H(A_x)$}
    & \multicolumn{4}{|c|}{\scriptsize$H(B)$}
    & \multicolumn{4}{|c|}{\scriptsize$H(C_z)$}
    & \multicolumn{4}{|c|}{\scriptsize$H(A_xB)$}
    & \multicolumn{4}{|c|}{\scriptsize$H(BC_z)$}
    & \multicolumn{4}{|c|}{\scriptsize$H(A_xBC_z)$}\\
\tiny{$x$ / $z$ / $xz$} & \text{\ }&\tiny{0}&\tiny{1}&\textbf{}&
\textbf{}&\textbf{}&\tiny{0}&\textbf{}& \text{\
}&\tiny{0}&\tiny{1}&\textbf{}& \text{\
}&\tiny{0}&\tiny{1}&\textbf{}&
\text{\ }&\tiny{0}&\tiny{1}&\textbf{}&
\tiny{00}&\tiny{01}&\tiny{10}&\tiny{11}
\\
\hline \textbf{1}& \text{\ }&\small{$-1$}&\small{0}&\textbf{}&
\textbf{}&\textbf{}&\small{$-1$}&\textbf{}& \text{\
}&\small{0}&\small{0}&\textbf{}& \text{\
}&\small{$1$}&\small{0}&\textbf{}&
\text{\ }&\small{0}&\small{0}&\textbf{}&
\small{0}&\small{0}&\small{0}&\small{0}
\\
\hline \textbf{2}& \text{\ }&\small{0}&\small{0}&\textbf{}&
\textbf{}&\textbf{}&\small{0}&\textbf{}& \text{\
}&\small{0}&\small{0}&\textbf{}& \text{\
}&\small{$1$}&\small{0}&\textbf{}&
\text{\ }&\small{0}&\small{0}&\textbf{}&
\small{$-1$}&\small{0}&\small{0}&\small{0}
\\
\hline \textbf{3}& \text{\ }&\small{$1$}&\small{0}&\textbf{}&
\textbf{}&\textbf{}&\small{0}&\textbf{}& \text{\
}&\small{$1$}&\small{0}&\textbf{}& \text{\
}&\small{0}&\small{0}&\textbf{}&
\text{\ }&\small{0}&\small{0}&\textbf{}&
\small{$-1$}&\small{0}&\small{0}&\small{0}
\\
\hline \textbf{4}& \text{\ }&\small{0}&\small{0}&\textbf{}&
\textbf{}&\textbf{}&\small{$1$}&\textbf{}& \text{\
}&\small{0}&\small{0}&\textbf{}& \text{\
}&\small{$-1$}&\small{0}&\textbf{}&
\text{\ }&\small{$-1$}&\small{0}&\textbf{}&
\small{$1$}&\small{0}&\small{0}&\small{0}
\\
\hline \textbf{5}& \text{\ }&\small{0}&\small{$1$}&\textbf{}&
\textbf{}&\textbf{}&\small{0}&\textbf{}& \text{\
}&\small{$1$}&\small{0}&\textbf{}& \text{\
}&\small{$1$}&\small{$-1$}&\textbf{}&
\text{\ }&\small{0}&\small{0}&\textbf{}&
\small{$-1$}&\small{0}&\small{0}&\small{0}
\\
\hline \textbf{6}& \text{\ }&\small{0}&\small{$1$}&\textbf{}&
\textbf{}&\textbf{}&\small{0}&\textbf{}& \text{\
}&\small{0}&\small{$1$}&\textbf{}& \text{\
}&\small{$1$}&\small{$-1$}&\textbf{}&
\text{\ }&\small{$1$}&\small{$-1$}&\textbf{}&
\small{$-1$}&\small{0}&\small{0}&\small{0}
\\
\hline \textbf{7}& \text{\ }&\small{$1$}&\small{0}&\textbf{}&
\textbf{}&\textbf{}&\small{0}&\textbf{}& \text{\
}&\small{$1$}&\small{0}&\textbf{}& \text{\
}&\small{0}&\small{0}&\textbf{}&
\text{\ }&\small{0}&\small{0}&\textbf{}&
\small{0}&\small{$-1$}&\small{$-1$}&\small{$1$}
\\
\hline \textbf{8}& \text{\ }&\small{$1$}&\small{0}&\textbf{}&
\textbf{}&\textbf{}&\small{0}&\textbf{}& \text{\
}&\small{$1$}&\small{0}&\textbf{}& \text{\
}&\small{$-1$}&\small{$1$}&\textbf{}&
\text{\ }&\small{0}&\small{0}&\textbf{}&
\small{$-1$}&\small{$1$}&\small{0}&\small{$-1$}
\\
\hline \textbf{9}& \text{\ }&\small{$1$}&\small{0}&\textbf{}&
\textbf{}&\textbf{}&\small{0}&\textbf{}& \text{\
}&\small{$1$}&\small{0}&\textbf{}& \text{\
}&\small{$-1$}&\small{$1$}&\textbf{}&
\text{\ }&\small{$-1$}&\small{$1$}&\textbf{}&
\small{$1$}&\small{$-1$}&\small{$-1$}&\small{0}
\\
\hline \textbf{10}& \text{\ }&\small{0}&\small{0}&\textbf{}&
\textbf{}&\textbf{}&\small{0}&\textbf{}& \text{\
}&\small{0}&\small{0}&\textbf{}& \text{\
}&\small{$1$}&\small{0}&\textbf{}&
\text{\ }&\small{$1$}&\small{0}&\textbf{}&
\small{$-1$}&\small{$-1$}&\small{$-1$}&\small{$1$}
\\
\hline
\end{tabular}
\caption{All classes of entropic bilocality inequalities. We have
listed the coefficients of one inequality in each row, and all
inequalities are of the form $\leq 0$.} \label{bilocalineqstable}
\end{table*}

\begin{figure}[t!]
\begin{center}
\includegraphics[width=0.8\linewidth]{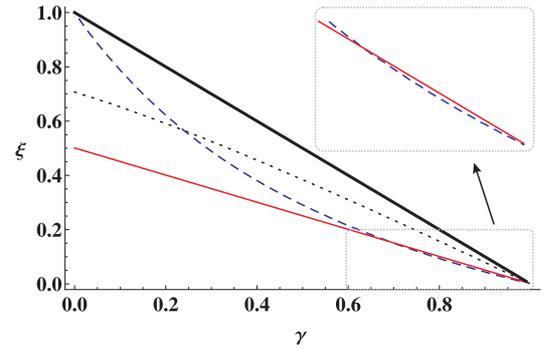}
\caption{(Color online) Parameter space of the
family~(\ref{nbbox}). The red dot-dashed line delimits the region
where~(\ref{AC}) and~(\ref{AAC}), and therefore~(\ref{nbbox}), are
local as witnessed by the CHSH inequality. The region above the
blue dashed line contains violations of the entropic bilocality
inequality \#7 from Table~\ref{bilocalineqstable}. The black line
delimits the region where~(\ref{AC}) and~(\ref{AAC}), and
therefore~(\ref{nbbox}), have a quantum realization with a single
$3$-qubit source. There is a small region (inset) in which our
inequality detects that~(\ref{nbbox}) is not bilocal, although the
box is tripartite local.} \label{fig:biviolation}
\end{center}
\end{figure}

\paragraph*{Looking for quantum violations.}

We now consider the quantum case. Instead of sending out
independent hidden variables $\lambda_1$ and $\lambda_2$, the two
sources now emit entangled quantum states. We take these to be
given by generic partially entangled two-qubit states
$$
\cos{\theta_k}\ket{00}+\sin{\theta_k}e^{i\phi_k}\ket{11}
$$
with $k=1,2$ indexing the two sources. Then $A$ and $C$ receive
one qubit each, while $B$ receives two. Upon choosing the
two-qubit measurement of $B$ to be in the Bell basis and
numerically optimizing over all projective measurements for A and
C, we have not been able to find any quantum violation of any of
our $52$ entropic bilocality inequalities.

On the other hand, it is not difficult to design some general
no-signaling boxes $P(a,b,c|x,z)$ which satisfy~(\ref{binosig}),
but violate some of our entropic inequalities. This applies for
example to the family of boxes, for parameters
$\xi,\gamma\in[0,1]$,
\begin{align}
\label{nbbox}
& P^{NB}(a,b,c|x,z)  \\
& = \frac{1}{8} \left( 1 + \xi (-1)^{{a}\oplus {b} \oplus
{c}\oplus xz} + (1-\xi-\gamma) (-1)^{{a}\oplus{b}\oplus{c}}
\right) \nonumber
\end{align}
This box can be understood as follows. The two outcomes
$b\in\{0,1\}$ both occur with probability $1/2$; if $b=0$, then
this creates between A and C the bipartite box
\begin{equation}
\label{AC} \xi P^{PR}+\gamma P^{c}+(1-\xi-\gamma)P^{w}
\end{equation}
and if $b=1$, then the resulting box between A and C is \beq
\label{AAC} \xi P^{APR}+\gamma P^{Ac}+(1-\xi-\gamma)P^{w}, \eeq
where $P^{APR}$ and $P^{Ac}$ stand for an ``anti-PR-box'', defined
as in~(\ref{isodist}) with $C=-1$, and classical
anti-correlations, respectively.

Depending on the value of $\xi$ and $\gamma$, the box $P^{NB}$ can
be tripartite local, tripartite nonlocal but tripartite quantum,
or post-quantumly tripartite nonlocal. $P^{NB}$
satisfies~(\ref{binosig}) since the bipartite marginal between A
and C is pure white noise, $P^{NB}(a,c|x,z)=P^{w}(a,c|x,z)=1/4$.

$P^{NB}$ violates some of our entropic bilocality inequalities. We
focus on the inequality 7 from Table~\ref{bilocalineqstable}.
Fig.~\ref{fig:biviolation} shows the region of violations as a
function of the pararameters $\xi$ and $\gamma$. Interestingly,
even in the region where $P^{NB}$ is local as a tripartite box,
the entropic bilocality inequality can be violated. This witnesses
that in those parameter ranges, $P^{NB}$ cannot be written in the
form~(\ref{bilocaldecom}), although the box is local
and~(\ref{binosig}) is satisfied.

\section{Conclusion}
\label{sec:conclusion}

In this work, we have exploited our general
framework~\cite{entropicpaper2} for deriving entropic
contextuality inequalities and entropic Bell inequalities. The
standard methods of computational geometry (like Fourier-Motzkin
elimination), which have been applied widely to the computation of
tight Bell inequalities, can be used for the computation of tight
(Shannon-type) entropic contexuality inequalities (Bell
inequalities). Following~\cite{AbramBrand} and related work, our
framework also treats nonlocality as a special case of
contextuality.

We also had shown in~\cite{entropicpaper2} that the family of
chained entropic inequalities derived by Braustein and
Caves~\cite{CHSHentropic} are the only non-trivial facets in the
appropriate scenario (Theorem~\ref{polygonthm}). Given that, we
have investigated quantum and more general violations of these
inequalities both in the Bell scenario case and in the
contextuality scenario case. Using a model of joint detection for
compatible observables, we noticed that quantum violations of a
certain entropic contextuality inequality exist for any positive
detection efficiency. Furthermore, we have fully characterized the
entropic inequalities for the simplest bilocality scenario. The
entropic bilocality inequalities can be violated by general
no-signalling correlations respecting the obvious
condition~(\ref{binosig}), but no quantum violations could be
found in terms of two independent sources of entangled quantum
states.

We have asked the question what the violation of an entropic
contextuality inequality or entropic Bell inequality might be a
resource for. In the case of the CHSH scenario, we have approached
this by noting some superficial similarities with the distillation
of nonlocality. We have speculated that the violation of the
entropic CHSH inequality may be a sufficient condition for the
possibility of distillation. If this would turn out to be true, it
may be very useful, since in general it is very difficult to
decide which no-signaling boxes are distillable and which ones are
not.

With our general framework, many more possibilities can be
explored. For example, the principle of information causality
(IC)~\cite{IC}, which has been proposed in order to understand the
implausible consequences of super-quantum correlations, is also an
entropic inequality. It can be shown that the inequality defining
IC is, up to symmetries, the only non-trivial facet of an entropic
cone defined by the IC scenario~\cite{preparation}. In principle,
modifying the scenario will therefore let us derive many other
IC-like principles as entropic inequalities, in particular some
with a multipartite flavor.

Since practical computations with entropic cones are extremely
demanding, it would be helpful to have more efficient algorithms
for the computation of facets of polyhedral cones in order to
characterize all (Shannon-type) entropic Bell inequalities for
bipartite scenarios with more observables per party, multipartite
Bell scenarios, and IC scenarios. To identify nonlinear problems
that turn to be linear in terms of entropies, as we have done with
bilocality, may also be an appealing line of future research.

\section*{Postscript}

While finishing this paper, the work~\cite{KRK} has appeared,
which also contains some of our results of
section~\ref{subsec:entropic_context}.

\begin{acknowledgments}
R.\,C.~would like to thank S.\,P.~Walborn for introducing him to
the Braunstein-Caves inequality. R.\,C.~was funded by the QESSENCE
project, T.\,F.~by the EU STREP QCS.
\end{acknowledgments}

\end{document}